\documentclass[a4paper]{article}

\usepackage{color}
\usepackage{authblk}
\usepackage{amsmath}
\usepackage{amsthm}
\usepackage{graphicx}
\usepackage{epstopdf}                  
\usepackage[ruled,linesnumbered,algo2e,vlined]{algorithm2e}
\usepackage{enumitem}
\usepackage{subcaption}
\usepackage{amssymb}

\newtheorem{definition}{Definition}
\newtheorem{theorem}{Theorem}
\newtheorem{lemma}{Lemma}

\newcommand{\etal}{\textit{et~al.}}

\newcommand{\pair}[2]{\langle #1, #2 \rangle}
\newcommand{\ST}{{\textit{\textbf{S}}}}  %% two dimensional 
\newcommand{\TT}{{\textit{\textbf{T}}}}  %% two dimensional text
\newcommand{\PP}{{\textit{\textbf{P}}}}  %% two dimensional pattern

\newcommand{\ZZ}[2]{\textit{Z}_{#1,#2}}
\newcommand{\ser}[1]{\textit{serial}(#1)}

\SetKw{Break}{break}
\SetKw{OR}{or}
\SetKw{Continue}{continue}
\SetKwProg{Fn}{Functi
	on}{}{end}\SetKwFunction{Dueling}{Dueling}%
\SetKwProg{Fn}{Function}{}{end}\SetKwFunction{Witness}{Witness}%
\SetKwProg{Fn}{Function}{}{end}\SetKwFunction{DuelingStage}{DuelingStage}%
\SetKwProg{Fn}{Function}{}{end}\SetKwFunction{SweepingStage}{SweepingStage}%
\SetKwProg{Fn}{Function}{}{end}\SetKwFunction{WitnessTwoDim}{Witness}%
\SetKwProg{Fn}{Function}{}{end}\SetKwFunction{DuelingStageTwoDim}{DuelingStage}%
\SetKwProg{Fn}{Function}{}{end}\SetKwFunction{BottomUp}{BottomUp}%
\SetKwProg{Fn}{Function}{}{end}\SetKwFunction{TopDown}{TopDown}%
\SetKwProg{Fn}{Function}{}{end}\SetKwFunction{Sweep}{Sweep}%
\SetKwProg{Fn}{Function}{}{end}\SetKwFunction{SweepImproved}{SweepImproved}%
\SetKwProg{Fn}{Function}{}{end}\SetKwFunction{OneDimensionalSweep}{OneDimensionalSweep}%
\SetKwProg{Fn}{Function}{}{end}\SetKwFunction{SweepWithinColumn}{SweepWithinColumn}%
\SetAlgoLongEnd

\newcommand{\WIT}[1]{\mathit{WIT\!}_{#1}}
\newcommand{\Prev}[1]{\mathit{Lmax}_{#1}}
\newcommand{\Next}[1]{\mathit{Lmin}_{#1}}

\newcommand{\imax}{i_{\mathit{max}}}
\newcommand{\imin}{i_{\mathit{min}}}

\graphicspath{ {figures/} }

\title{Duel and sweep algorithm for order-preserving pattern matching}

\author[1]{Davaajav Jargalsaikhan}
\author[1]{Diptarama}
\author[1]{Ryo Yoshinaka}
\author[1]{Ayumi Shinohara}

\affil[1]{Graduate School of Information Sciences, Tohoku University\\
	6-6-05 Aramaki Aza Aoba, Aoba-ku, Sendai, Japan\\
	\texttt{\{davaajav@shino., diptarama@shino., ry@, ayumi@\}ecei.tohoku.ac.jp}}

\begin{document}

\maketitle

\begin{abstract}
	% -*- coding:utf-8 -*-

Given a text $T$ and a pattern $P$ over alphabet $\Sigma$,
the classic exact matching problem searches for all occurrences of pattern $P$ in text $T$.
Unlike exact matching problem, \emph{order-preserving pattern matching} (OPPM) considers the relative order of elements, rather than their real values.
In this paper, we propose an efficient algorithm for OPPM problem using the ``duel-and-sweep'' paradigm.
Our algorithm runs in $O(n + m\log m)$ time in general and
$O(n + m)$ time under an assumption that the characters in a string can be sorted in linear time with respect to the string size.
We also perform experiments and show that our algorithm is faster that KMP-based algorithm.
Last, we introduce the two-dimensional order preserved pattern matching
and give a duel and sweep algorithm that runs in $O(n^2)$ time for duel stage and $O(n^2 m)$ time for sweeping time
with $O(m^3)$ preprocessing time.

\end{abstract}

%-*- coding: utf-8 -*-

\section{Introduction}

The exact string matching problem is one of the most widely studied problems. 
Given a text and a pattern, the exact matching problem searches for all occurrences positions of pattern in the text.
Motivated by low level image processing, the two-dimensional exact matching problem has been extensively studied in recent decades.
Given a text $T$ of size $n \times n$ and a pattern $P$ of size $m \times m$ over alphabet $\Sigma$ of size $\sigma = |\Sigma|$, the exact matching problem on two-dimensional strings searches for all occurrence positions of $P$ in $T$.
Bird~\cite{bird1977two} and Baker~\cite{baker1978technique} proposed two-dimensional exact matching using dictionary matching algorithm and Amir and Farach~\cite{amir1992two} proposed an algorithm that uses suffix trees.
These algorithms require total ordering from the alphabet and run in $O (n ^ 2 \log \sigma)$ time with $O (m ^ 2 \log \sigma)$ preprocessing time.
Amir~\etal~\cite{amir1994alphabet} also proposed alphabet independent approach to the problem that runs in $O(m^2 \log \sigma)$ preprocessing time and $O(n^2)$ matching time.

Unlike the exact matching problem, \emph{order-preserving pattern matching} (OPPM) considers the relative order of elements, rather than their real values.
Order-preserving matching has gained much interest in recent years, due to its applicability in problems where the relative order is compared, rather than the exact value, such as share prices in stock markets, weather data or musical notes.

Kubica~\etal~\cite{kubica2013linear} and Kim~\etal~\cite{kim2014order} proposed a solution based on KMP algorithm.
These algorithms address the one-dimensional OPPM problem and have time complexity of $O (n + m \log m)$.
Cho~\etal~\cite{cho2015fast} brought forward another algorithm based on the Horspool's algorithm that uses $q$-grams, which was proven to be experimentally fast.
Crochemore~\etal~\cite{SPIRE_Crochemore_2013} proposed data structures for OPPM.
On the other hand, Chhabra and Tarhio~\cite{SEA_Chhabra_2014}, Faro and K\"{u}lekci~\cite{faro2016efficient} proposed filtration methods which practically fast.
Moreover, faster filtration algorithms by using SIMD (Single Instruction Multiple Data) instructions
were proposed by Cantone~\etal~\cite{ref:PSC_Cantone}, Chhabra~\etal~\cite{ref:PSC_Chhabra} and Ueki~\etal~\cite{ueki2016fast}.
They showed that SIMD instructions are efficient in speeding up their algorithms.

In this paper, we propose an algorithm that based on dueling technique~\cite{vishkin1985optimal} for OPPM.
Our algorithm runs in $O(n + m\log m)$ time which is as fast as KMP based algorithm.
Moreover, we perform experiments those compare the performance of our algorithm with the KMP-based algorithm.
The experiment results show that our algorithm is faster that KMP-based algorithm.
Last, we introduce the two-dimensional order preserved pattern matching
and give a duel and sweep algorithm that runs in $O(n^2)$ time for duel stage and $O(n^2 m)$ time for sweeping time
with $O(m^3)$ preprocessing time.
To the best of our knowledge, our solution is the first to address the two-dimensional order preserving patern matching problem.

The rest of the paper is organized as follows. In Section~\ref{sec:prelim}, we give preliminaries on the problem.
In Section~\ref{sec:one dimension}, we describe the algorithm for OPPM problem.
In Section~\ref{sec:experiment} we will show some experiment results those compare the performance of our algorithm with the KMP-based algorithm.
In Section~\ref{sec:two dimension}, we extend the algorithm and describe the method for the two-dimensional OPPM problem.
In Section~\ref{sec:conclusion}, we conclude our work and discuss future work. 
% -*- coding:utf-8 -*-
\section{Preliminaries} \label{sec:prelim}

We use $\Sigma$ to denote an alphabet of integer symbols such that the comparison of any two symbols can be done in constant time. $\Sigma^*$ denotes the set of strings over the alphabet $\Sigma$.
For a string $S \in \Sigma^*$, we will denote $i$-th element of $S$ by $S[i]$ and a substring of $S$ that starts at the location $i$ and ends at the location $j$ as $S[i\!:\!j]$. 
We say that two strings $S$ and $T$ of equal length $n$ are \emph{order-isomorphic}, written $S \approx T$, if $S[i] \leq S[j] \Longleftrightarrow T[i] \leq T[j]$ for all $1 \leq i, j \leq n$.
For instance, $(12, 35, 5) \approx (25, 30, 21) \not\approx (11, 13, 20)$.

In order to check order-isomorphism of two strings, Kubica~\etal~\cite{kubica2013linear} introduced~\footnote{Similar arrays $\textit{Prev}_S$ and $\textit{Next}_S$ are introduced in~\cite{hasan2015order}.} useful arrays $\Prev{S}$ and $\Next{S}$ defined by
\begin{align}
	\Prev{S}[i]=j \mbox{ if } S[j]=\max_{k < i} \{S[k] \mid  S[k] \le S[i] \}, \\
	\Next{S}[i]=j \mbox{ if } S[j]=\min_{k < i} \{S[k] \mid  S[k] \ge S[i] \}.
\end{align}
We use the rightmost (largest) $j$ if there exist more than one such $j$. 
If there is no such $j$ then we define $\Next{S}[i] = 0$ and $\Prev{S}[i] = 0$, respectively.
From the definition, we can easily observe the following properties.
\begin{align}
	S[\Prev{S}[i]] = S[i]\quad \Longleftrightarrow\quad S[i] = S[\Next{s}[i]], \label{prop:lmaxmineq} \\
	S[\Prev{S}[i]] < S[i]\quad \Longleftrightarrow\quad S[i] < S[\Next{s}[i]]. \label{prop:lmaxminineq}
\end{align}
\begin{lemma}[\cite{kubica2013linear}]
	For a string $S$, let $sort(S)$ be the time required to sort the elements of $S$.
	$\Prev{S}$ and $\Next{S}$ can be computed in $O(sort(S) + |S|)$ time.
\end{lemma}
Thus, $\Prev{S}$ and $\Next{S}$ can be computed in $O(|S|\log |S|)$ time in general.
Moreover, the computation can be done in $O(|S|)$ time 
under a natural assumption~\cite{kubica2013linear} that the characters of $S$ are elements of the set $\{1,\ldots,|S|^{O(1)}\}$.
By using $\Prev{S}$ and $\Next{S}$, order-isomorphism of two strings can be decided as follow.
\begin{lemma}[\cite{cho2015fast}]\label{lem:op}
	For two strings $S$ and $T$ of length $n$,
	assume that $S[1\!:\!i] \approx T[1\!:\!i]$ for some $i < n$. % and we have calculated $\Next{S}$ and $\Prev{S}$.
	Let $\imax = \Prev{S}[i+1]$ and $\imin = \Next{S}[i+1]$.
	Then $S[1\!:\!i+1] \approx T[1\!:\!i+1]$ if and only if either of the following two conditions holds.
	\begin{align}
	S[\imax] = S[i+1] = S[\imin]\ \wedge\ T[\imax] = T[i+1] = T[\imin], \label{eq:opeq}\\
	S[\imax] < S[i+1] < S[\imin]\ \wedge\ T[\imax] < T[i+1] < T[\imin]. \label{eq:opineq}
	\end{align}
	We omit the corresponding equalities/inequalities if $\imax=0$ or $\imin=0$.
\end{lemma}

Hasan~\etal~\cite{hasan2015order} proposed a modification to Z-function, which Gusfield~\cite{gusfield1997algorithms} defined for ordinal pattern matching, to make it useful from the order-preserving point of view. 
For a string $S$, the \emph{(modified) Z-array} of $S$ is defined by
\[Z_S[i] = \max_{1 \le j \le |S| - i + 1}\{j \mid S[1:j]\approx S[i:i+j-1] \} \mbox{\quad for each } 1 \leq i \leq |S|.\]
In other words, $Z_S[i]$ is the length of the longest substring of $S$ that starts at position $i$ and is order-isomorphic with some prefix of $S$. An example of Z-array is illustrated in Table \ref{table:Z-array}.

\begin{table}[t]
	\centering
		\caption{Z-array of a string $S = (18, 22, 12, 50, 10, 17)$.
			For instance, $Z_S[3] = 3$ because $S[1\!:\!3] = (18, 22, 12) \approx (12, 50, 10) = S[3\!:\!5]$ and $S[1\!:\!4] = (18, 22, 12, 50) \not\approx (12, 50, 10, 17) = S[3\!:\!6]$. $\Prev{S}$ and $\Next{S}$ are also shown.　} 
		\label{table:Z-array}
		\setlength{\tabcolsep}{8pt}
		\begin{tabular}{|c|c|c|c|c|c|c|}
			\hline
			    & $1$ & $2$ & $3$ & $4$ & $5$ & $6$ \\
			\hline
			$S$ & $18$ & $22$ & $12$ & $50$ & $10$ & $17$ \\
			\hline
			$Z_S$ & $6$ & $1$ & $3$ & $1$ & $2$ & $1$ \\
			\hline
			$\Prev{S}$ & $0$ & $1$ & $0$ & $2$ & $0$ & $3$ \\
			\hline
			$\Next{S}$ & $0$ & $0$ & $1$ & $0$ & $3$ & $1$ \\
			\hline
		\end{tabular}
\end{table}

\begin{lemma} (\cite{hasan2015order}) \label{lemma:Z-array} 
	For a string $S$, Z-array $Z_S$ can be computed in $O(|S|)$ time, assuming that $\Prev{S}$ and $\Next{S}$ are already computed.
\end{lemma}

Note that in their original work, Hasan \etal~\cite{hasan2015order} assumed that each character in $S$ is distinct.
However, we can extend their algorithm by using Lemma~\ref{lem:op} to verify order-isomorphism even when $S$ contains duplicate characters.

%-*- coding: utf-8 -*-

\section{One-dimensional order-preserving matching} \label{sec:one dimension}

In this section, we will propose an algorithm for one-dimensional OPPM using the ``duel-and-sweep" paradigm~\cite{amir1994alphabet}. 
In the dueling stage, all possible pairs of candidates ``duel'' with each other.
The surviving candidates are further pruned during the sweeping stage, leaving the candidates that are order-isomorphic with the pattern. 
Prior to the dueling stage, the pattern is preprocessed to construct a \emph{witness table} that contains \emph{witness pairs} for all possible offsets.

\begin{definition} [1d-OPPM problem]
The one-dimensional order-preserving matching problem is defined as follows,
\begin{description}[topsep=0pt,parsep=0pt,partopsep=0pt]
	\item[Input:] A text $T \in \Sigma^*$ of length $n$ and a pattern $P \in \Sigma^*$ of length $m$,
	\item[Output:] All occurrences of substrings of $T$ that are order-isomorphic with $P$.
\end{description}
\end{definition}

\subsection{Pattern preprocessing}

Let $a > 0$ be an integer such that when $P$ is superimposed on itself with the offset $a$, the overlap regions are not order-isomorphic. We say that a pair $\pair{i}{j}$ of locations is  \emph{a witness pair for the offset $a$} if either of the following holds:
\begin{itemize}[topsep=3pt]
	\item $P[i] = P[j] \text{ and } P[i+a] \ne P[j+a]$,
	\item $P[i] > P[j] \text{ and } P[i+a] \le P[j+a]$,
	\item $P[i] < P[j] \text{ and } P[i+a] \ge P[j+a]$.
\end{itemize}
Next, we describe how to construct a \emph{witness table} for $P$, that stores witness pairs for all possible offsets $a$ $(0 < a < m)$.
For the one-dimensional problem, the witness table $\WIT{P}$ is an array of length $m-1$, such that $\WIT{P}[a]$ is a witness pair for offset $a$.
In the case when there are multiple witness pairs for offset $a$, we take the pair $\pair{i}{j}$ with the smallest value of $j$ and $i < j$.
When the overlap regions are order-isomorphic for offset $a$, which implies that no witness pair exists for $a$, we express it as $\WIT{P}[a] = \pair{m+1}{m+1}$. 

\begin{lemma} \label{lemma:Z_p computation}
	For a pattern $P$ of length $m$, we can construct $\WIT{P}$ in $O(m)$ time assuming that $Z_P$ is already computed.
\end{lemma} 
\begin{proof}
Remind that $Z_P[k]$ is the length of the longest prefix of $P[k\!:\!m]$ that is order-isomorphic with a prefix of $P$. 
For each $1 < k < m$, we have two cases.
\begin{description}[topsep=5pt,itemsep=3pt]
	\item[Case 1] $Z_P[k] = m - k + 1$ : Since $P[1\!:\!m - k + 1] \approx P[k\!:\!m]$, there is no witness pair for offset $k - 1$.
	\item[Case 2] $Z_P[k] < m - k + 1$ : 
	Let $j_k = Z_P[k] + 1$, $\imax = \Prev{P}[j_k]$, and $\imin = \Next{P}[j_k]$.
	Then $P[1\!:\!j_k-1]\approx P[k\!:\!k+j_k-2]$ and $P[1\!:\!j_k]\not \approx P[k\!:\!k+j_k-1]$, by the definition of $Z_P[k]$.
	By Lemma~\ref{lem:op}, neither condition~(\ref{eq:opeq}) nor (\ref{eq:opineq}) holds.
	If ${P[\imax]} = P[j_k]$ then $P[j_k] = P[\imin]$ by property~(\ref{prop:lmaxmineq}), so that
	\begin{align}
		P[k + \imax - 1] \ne P[k+j_k-1]\ \vee\ P[k+j_k-1] \ne P[k + \imin - 1] \label{cond:eq}
	\end{align}
	holds by condition~(\ref{eq:opeq}). Otherwise, i.e. ${P[\imax]} < P[j_k]$, we have $P[j_k] < P[\imin]$ by property~(\ref{prop:lmaxmineq}), so that
	\begin{align}
		P[k + \imax - 1] \ge P[k+j_k-1]\ \vee P[k+j_k-1]\ \ge P[k + \imin - 1] \label{cond:ineq}
	\end{align}
	holds by condition~(\ref{eq:opineq}).
	Therefore, $\pair{\imax}{j_k}$ is a witness pair if the leftside of condition~(\ref{cond:eq}) or (\ref{cond:ineq}) holds,
	and $\pair{\imin}{j_k}$ is a witness pair if rightside of condition~(\ref{cond:eq}) or (\ref{cond:ineq}) holds.
\end{description} 
Algorithm~\ref{alg:Witness} describes the procedure.
Clearly it runs in $O(m)$ time.
\end{proof}

\subsection{Dueling stage}

A substring of $T$ of length $m$ will be referred to as a \emph{candidate}. A candidate that starts at the location $x$ will be denoted by $T_x$.
Witness pairs are useful in the following situation. 
Let $T_x$ and $T_{x+a}$ be two overlapping candidates and $\pair{i}{j}$ be the witness pair for offset $a$.
Without loss of generality, we assume that $P[i] < P[j]$ and $P[i+a] > P[j + a]$. 
\begin{itemize}[topsep=3pt]
	\item If $T[x + a + i -1] > T[x + a + j -1]$, then $T_x \not\approx P$.
	\item If $T[x + a + i -1] < T[x + a + j -1]$, then $T_{x+a} \not\approx P$.
\end{itemize}
Based on this information, we can safely eliminate either candidate $T_x$ or $T_{x+a}$ without looking into other locations. This process is called \emph{dueling}. The procedure for the dueling is described in the Algorithm \ref{alg:Dueling}.

\begin{algorithm2e}[tbp]
	\Fn(\tcc*[h]{Construct the witness table $\WIT{P}$}){\Witness{}}{
		compute the Z-array $Z_P$ for the pattern $P$\;
		\For{$k=2$ \KwTo $m-1$}{
			$j = Z_P[k] + 1$\;
			\lIf{$j = m - k + 1$}{
					$\WIT{P}[k - 1] = \pair{m + 1}{m + 1}$}
				\ElseIf{$P[\Next{P}[j]] = P[j] = P[\Prev{P}[j]]$}{
				\If{$P[k+j-1] \ne P[k + \Prev{P}[j] - 1]$}{
					$\WIT{P}[k-1] = \pair{\Prev{P}[j]}{j}$\;
				}
				\lElse{
					$\WIT{P}[k-1] = \pair{\Next{P}[j]}{j}$}
			}\Else{
				
				\If{$P[k+j-1] \le P[k + \Prev{P}[j] - 1]$}{
					$\WIT{P}[k-1] = \pair{\Prev{P}[j]}{j}$\;
				}
				\lElse{
					$\WIT{P}[k-1] = \pair{\Next{P}[j]}{j}$}
			}
		}
	}
	\caption{Algorithm for constructing the witness table $\WIT{P}$}
	\label{alg:Witness}
\end{algorithm2e}

Next, we prove that the \emph{consistency} property is transitive.
Suppose $T_x$ and $T_{x+a}$ are two overlapping candidates.
We say that $T_x$ and $T_{x+a}$ are \emph{consistent} with respect to $P$ if $P[1\!:\!m-a] \approx P[a+1\!:\!m]$.
Candidates that do not overlap are trivially consistent.

\begin{lemma} \label{lemma:consistency 1d}
	For any $a$ and $a'$ such that $0 < a < a + a' < m$, let us consider three candidates $T_{x}$, $T_{x + a}$, and $T_{x + a + a'}$. If $T_x$ is consistent with $T_{x+a}$ and $T_{x+a}$ is consistent with $T_{x + a + a'}$, then $T_x$ is consistent with $T_{x + a+ a'}$.
\end{lemma}

\begin{proof}
	Since $T_x$ is consistent with $T_{x+a}$, it follows that $P[1\!:\!m - a] \approx P[a+1\!:\!m]$, so that $P[a' + 1\!:\!m - a] \approx P[(a + a')+1\!:\!m]$.
	Moreover, since $T_{x+a}$ is consistent with $T_{x+a +a'}$, it follows that $P[1\!:\!m - a'] \approx P[a'+1\!:\!m]$, so that $P[1\!:\!m - a' - a] \approx P[a'+1\!:\!m - a]$.
	Thus, $P[1\!:\!m - (a + a')] \approx P[(a + a') + 1\!:\!m]$, which implies that $T_x$ is consistent with $T_{x+a+a'}$.
\end{proof}

During the dueling stage, the candidates are eliminated until all remaining candidates are pairwise consistent.
For that purpose, we can apply the dueling algorithm due to Amir~\etal~\cite{amir1994alphabet} developed for ordinal pattern matching.

\begin{lemma}[\cite{amir1994alphabet}] \label{lemma:duel 1D}
The dueling stage can be done in $O(n^2)$ time by using $\WIT{\PP}$.
\end{lemma}

\begin{algorithm2e}[t]
	\Fn(\tcc*[h]{Duel between candidates $T_x$ and $T_{x+a}$}){\Dueling{$T_x, T_{x + a}$}}{
		$\pair{i}{j} = \WIT{P}[a]$\;
		\If{$P[i] = P[j]$}{
			\leIf{$T[x + a + i -1] \ne T[x + a + j -1]$}{
				\bf{return} $T_{x+a}$\;
			}{
				\bf{return} $T_{x}$}
		}
		\If{$P[i] < P[j]$}{
			\leIf{$T[x + a + i -1] > T[x + a + j -1]$}{
				\bf{return} $T_{x+a}$\;
			}{
				\bf{return} $T_{x}$}
		} 
		\If{$P[i] > P[j]$}{
			\leIf{$T[x + a + i -1] < T[x + a + j -1]$}{
				\bf{return} $T_{x+a}$\;
			}{
				\bf{return} $T_{x}$}
		} 
	}
	\caption{Dueling}
	\label{alg:Dueling}
\end{algorithm2e}

\subsection{Sweeping stage}\label{sec:1DSweeping}
The goal of the sweeping stage is to prune candidates until all remaining candidates are order-isomorphic with the pattern.
Suppose that we need to check whether some surviving candidate $T_x$ is order-isomorphic with the pattern $P$.
It suffices to successively check the conditions~(\ref{cond:eq}) and (\ref{cond:ineq}) in Lemma~\ref{lem:op}, starting from the leftmost location in $T_x$.
If the conditions are satisfied for all locations in $T_x$, %then $T_x$ is oder-isomorphic with $P$. 
then $T_x \approx P$.
Otherwise,
$T_x \not\approx P$, and obtain a mismatch position $j$.

A naive implementation of the sweeping will result in $O(n^2)$ time.
However, if we take advantage of the fact that all the remaining candidates are pairwise consistent,
we can reduce the time complexity to $O(n)$ time.
Since the remaining candidates are consistent to each other, for the overlapping candidates $T_x$ and $T_{x + a}$,
the overlap region is checked only once if $T_x$ is order-isomorphic with the pattern $P$.
Otherwise, for a mismatch position $j$,
$T_{x+a}$ should be checked from position $j - a + 1$ of $T_{x+a}$,
because $P[a:j-1] \approx T_x[a\!:\!j-1] \approx T_{x+a}[1\!:\!j-a]$.
Algorithm~\ref{alg:1DSweep} describes the procedure for the sweeping stage.

\begin{algorithm2e}[t]
	\Fn{\SweepingStage{}}{
		\While{there are unchecked candidates to the right of $T_x$}{
			let $T_x$ be the leftmost unchecked candidate\;
			\eIf{there are no candidates overlapping with $T_x$}{
				\lIf{$T_x \not \approx P$}{
					eliminate $T_x$}
			}{
				let $T_{x+a}$ be the leftmost candidate that overlaps with $T_x$\;
				\lIf{$T_x \approx P$}{
					start checking $T_{x+a}$ from the location $m - a + 1$}
				\Else{
					let $j$ be the mismatch position\;
					eliminate $T_x$\;
					start checking $T_{x+a}$ from the location $j-a$\;
				}
			}
		}
	}
	\caption{The sweeping stage algorithm}
	\label{alg:1DSweep}
\end{algorithm2e}

\begin{lemma} \label{lemma:Sweep 1D}
	The sweeping stage can be completed in $O(n)$ time.
\end{lemma}

By Lemmas~\ref{lemma:Z_p computation}, \ref{lemma:duel 1D}, and \ref{lemma:Sweep 1D}, we summarize this section as follows.
\begin{theorem}
	The duel-and-sweep algorithm solves 1d-OPPM Problem in $O(n + m\log m)$ time.
	Moreover, the running time is $O(n+m)$ under the natural assumption that
	the characters of $P$ can be sorted in $O(m)$ time.
\end{theorem}
% -*- coding:utf-8 -*-
\section{Experiment} \label{sec:experiment}

\begin{figure}[t]
	\centering
	\begin{minipage}[t]{0.49\hsize}
		\centering
		\includegraphics[scale=0.47]{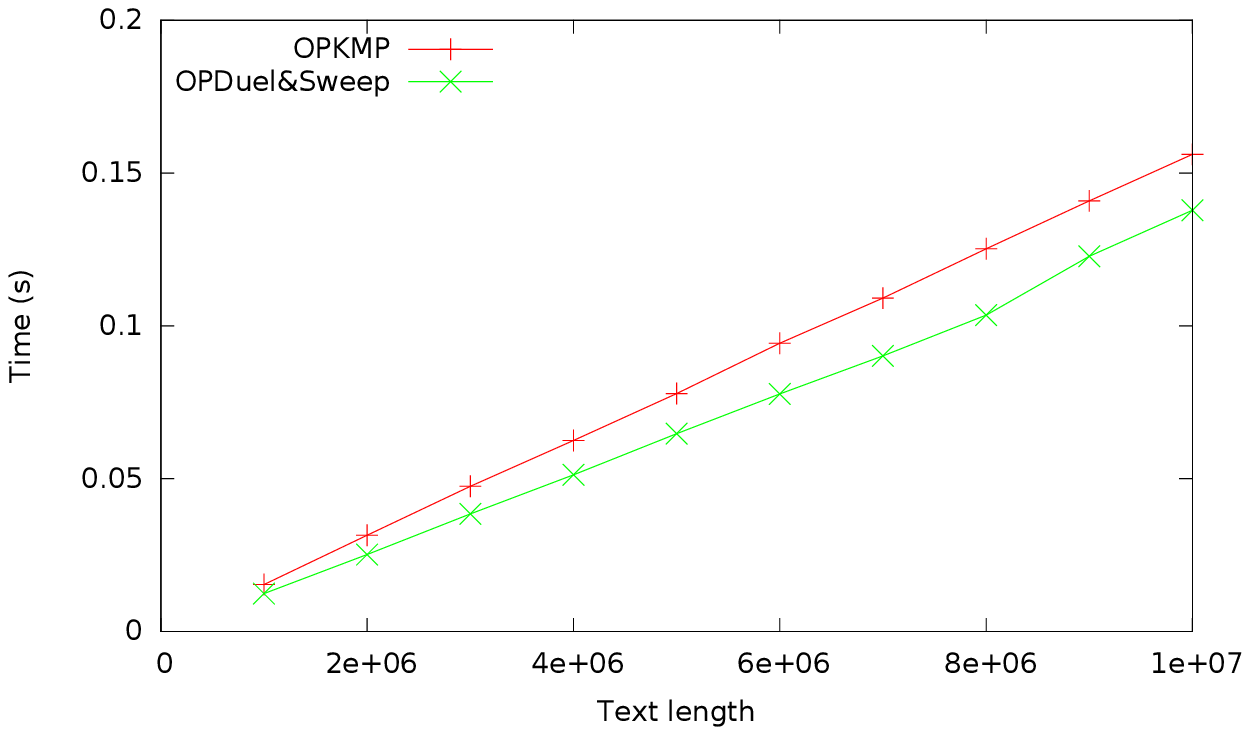}
		\ \ \ \scriptsize{(a)}
	\end{minipage}
	\begin{minipage}[t]{0.49\hsize}
		\centering
		\includegraphics[scale=0.47]{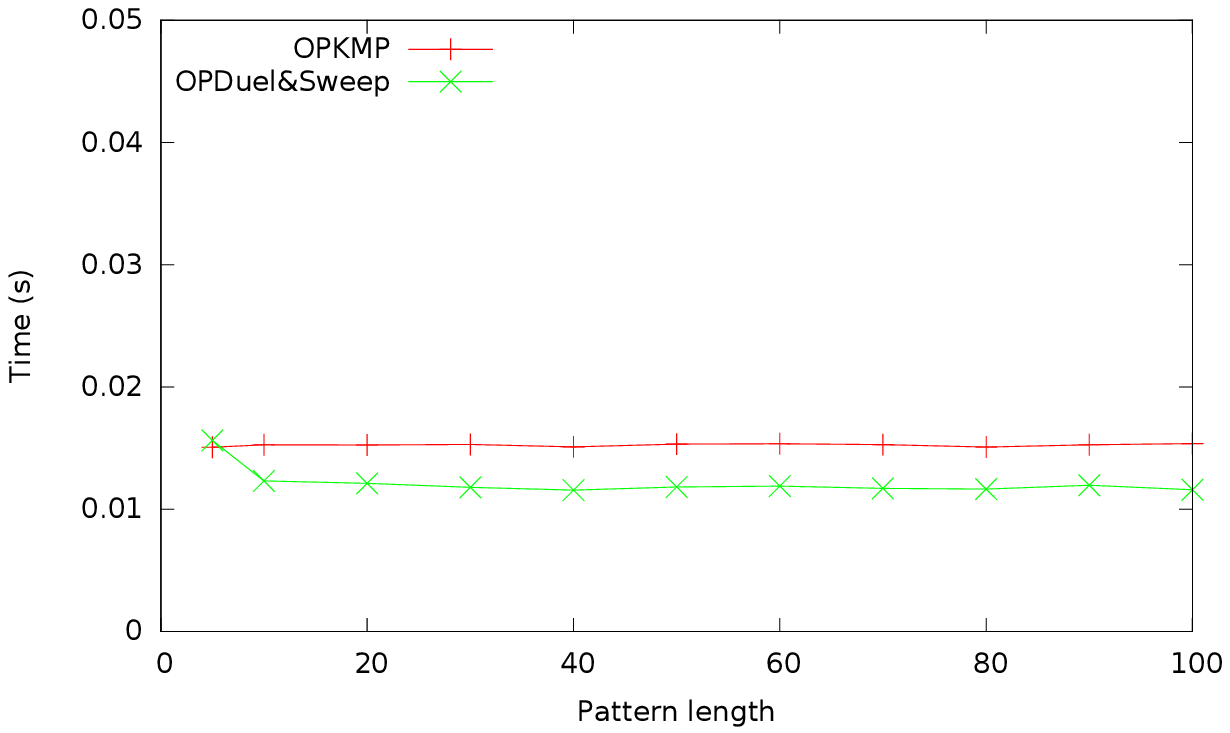}
		\ \ \ \scriptsize{(b)}
	\end{minipage}
	\caption{Running time of the algorithms with respect to (a)~text length, and (b)~pattern length.}
	\label{fig:ex_time}
\end{figure}

\begin{figure}[t]
	\centering
	\begin{minipage}[t]{0.49\hsize}
		\centering
		\includegraphics[scale=0.47]{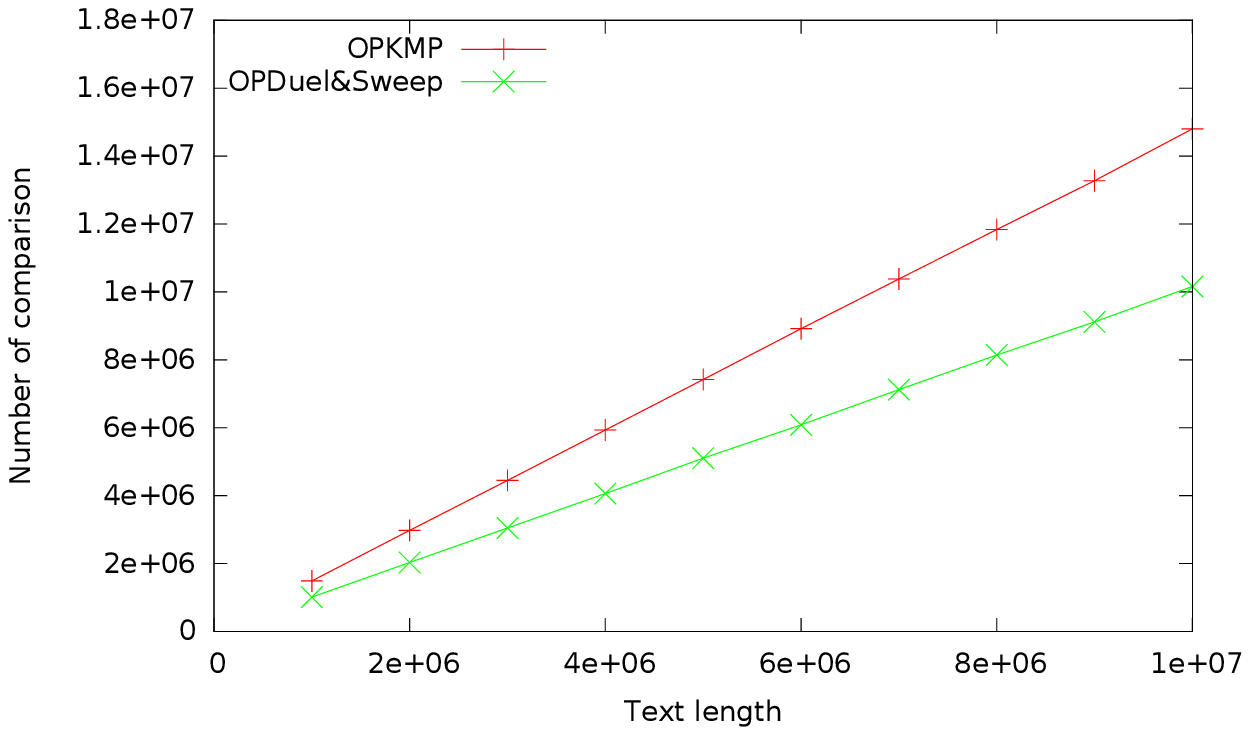}
		\ \ \ \scriptsize{(a)}
	\end{minipage}
	\begin{minipage}[t]{0.49\hsize}
		\centering
		\includegraphics[scale=0.47]{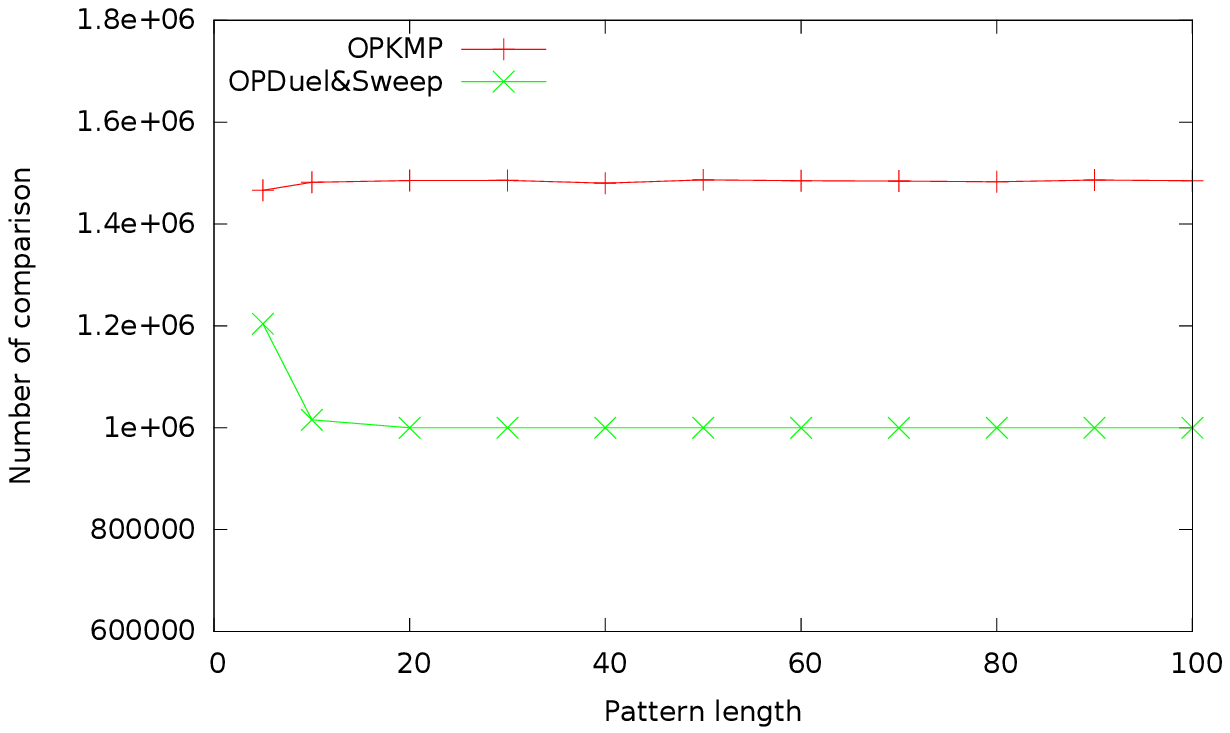}
		\ \ \ \scriptsize{(b)}
	\end{minipage}
	\caption{Number of comparisons in the algorithms with respect to (a)~text length, and (b)~pattern length.}
	\label{fig:ex_comp}
\end{figure}

In order to compare the performance of proposed algorithm with the KMP-based algorithm, we conducted experiments on 1d-OPPM problem.
We performed two sets of experiments.
In the first experiment, the pattern size $m$ is fixed to $10$, while the text size $n$ is changed from $100000$ to $1000000$.
In the second experiment, the text size $n$ is fixed to $1000000$ while the pattern size $m$ is changed from $m$ $5$ to $100$.
We measured the average of running time and the number of comparisons for $50$ repetitions on each experiment.
We used randomly generated texts and patterns with alphabet size $|\Sigma|=1000$.
Experiments are executed on a machine with Intel Xeon CPU E5-2609 8 cores 2.40 GHz, 256 GB memory, and Debian Wheezy operating system.

The results of our preliminary experiments are shown in Fig.~\ref{fig:ex_time} and Fig.~\ref{fig:ex_comp}.
We can see that our algorithm is better that KMP based algorithm in running time and number of comparison when the pattern size and text size are large.
However, our algorithm is worse when the pattern size is small, less than $10$.
%-*- coding: utf-8 -*-

\section{Two-dimensional order preserving pattern matching} \label{sec:two dimension}

In this section, we will discuss how to perform two-dimensional order preserving pattern matching (2d-OPPM).
Array indexing is used for two-dimensional strings, the horizontal coordinate $x$ increases from left to right and the vertical coordinate $y$ increases from top to bottom. $\ST[x,y]$ denotes an element of $\ST$ at position $(x,y)$ and $\ST[x\!:\!x + w - 1 , y\!:\!y + h - 1]$ denotes a substring of $\ST$ of size $w \times h$ with top-left corner at the position $(x, y)$.

We say that two dimensional strings $\ST$ and $\TT$ are \emph{order-isomorphic}, written $\ST \approx \TT$, if $\ST[i_x,i_y] \leq \ST[j_x,j_y] \Longleftrightarrow \TT[i_x,i_y] \leq \TT[j_x,j_y]$ for all $1 \leq i_x, j_x \leq w$ and $1 \leq i_y, j_y \leq h$.
For a simple presentation, we assume that both text and pattern are squares $(w=h)$ in this paper, but we can generalize it straightforwardly.
\begin{definition} [2d-OPPM problem]
	The two-dimensional order-preserving matching problem is defined as follows,
	\begin{description}[topsep=0pt,parsep=0pt,partopsep=0pt]
		\item[Input:] A text $\TT$ of size $n \times n$ and a pattern $\PP$ of size $m \times m$,
		\item[Output:] All occurrences of substrings of $\TT$ that are order-isomorphic with $\PP$.
	\end{description}
\end{definition}

Our approach is to reduce 2d-OPPM problem into 1d-OPPM problem, based on the following observation.
For two-dimensional string $\ST$, let $\ser{\ST}$ be a (one-dimensional) string which \emph{serializing} $\ST$ by traversing it
in the left-to-right/top-to-bottom order. We can easily verify the following lemma.

\begin{lemma} \label{lemma:serialized}
	$\ST \approx \TT$ if and only if $\ser{\ST} \approx \ser{\TT}$ for any $\ST$ and $\TT$.
\end{lemma}

\begin{theorem} \label{thm:2dOPPM}
	2d-OPPM problem can be solved in $O(n^2 m + m^2 \log{m})$. 
\end{theorem}

\begin{proof}
	For a fixed $1 \leq x \leq n -m + 1$, consider the substring $\TT[x:x+m-1, 1:n]$ and let $S_x = \ser{\TT[x:x+m-1, 1:n]}$.
	By Lemma~\ref{lemma:serialized}, $\PP$ occurs in $\TT$ at position $(x, y)$, i.e.
	$\PP \approx \TT[x:x+m-1, y:y+m-1]$ if and only if $\ser{\PP} \approx S_x[m(y-1)+1, m(y-1)+m^2]$.
	The positions $m(y-1)+1$ satisfying the latter condition can be found in $O(nm + m^2\log{m})$ time by 1d-OPPM algorithms, which we showed in Section~\ref{sec:one dimension} or KMP-based ones~\cite{kubica2013linear,kim2014order}, because $|S_x| = nm$ and $|\ser{\PP}|=m^2$.
	Because we need the preprocess for the pattern $\ser{\PP}$ only once, and execute the search in $S_x$ for each $x$, the result follows.
\end{proof}

In the rest of this paper, we try a direct approach to two-dimensional strings based on the duel-and-sweep paradigm, inspired by the work~\cite{amir1992two,Cole2014TwoDimParaMatch}.
A substring of $\TT$ of size $m \times m$ will be referred as a candidate. $\TT_{x,y}$ denotes a candidate with the top-left corner at $(x,y)$. 

\subsection{Pattern preprocessing}
For $0 \leq a < m$ and $-m < b < m$,
we say that a pair $\pair{(i_x, i_y)}{(j_x, j_y)}$ of locations is \emph{a witness pair for the offset $(a, b)$} if either of the following holds:
\begin{itemize}[topsep=3pt]
	\item $\PP[i_x, i_y] = \PP[j] \text{ and } \PP[i_x+a, i_y+b] \ne \PP[j_x, j_y]$,
	\item $\PP[i_x, i_y] > \PP[j] \text{ and } \PP[i_x+a, i_y+b] \le \PP[j_x, j_y]$,
	\item $\PP[i_x, i_y] < \PP[j] \text{ and } \PP[i_x+a, i_y+b] \ge \PP[j_x, j_y]$.
\end{itemize}
The \emph{witness table $\WIT{\PP}$ for pattern $\PP$} is a two-dimensional array of size $m \times (2m-1)$, where
$\WIT{\PP}[a, b]$ is a witness pair for the offset $(a, b)$. 
If the overlap regions are order-isomorphic when $\PP$ is superimposed with offset $(a, b)$, then no witness pair exists.
We denote it as $\WIT{\PP}[a,b] = \pair{(m+1, m+1)}{(m+1, m+1)}$.

We show how to efficiently construct the witness table $\WIT{\PP}$.
For $\PP$ and each $0 \leq a < m$, we define the \emph{Z-array} $\ZZ{\PP}{a}$ by 
\[\ZZ{\PP}{a}[i] = \max_{1 \le j \le |P_1| - i + 1}\{j \mid P_1[1:j]\approx P_2[i:i+j-1] \} \mbox{\  for each } 1 \leq i \leq |P_1|, \]
where $P_1 = \ser{\PP[1\!:\!m-a, 1\!:\!m]}$, $P_2 = \ser{\PP[a + 1\!:\!m, 1\!:\!m]}$, and $|P_1| = |P_2| = m(m-a)$.

\begin{lemma} \label{lemma:Wit-2d, each cell}
	For arbitrarily fixed $a \geq 0$, we can compute the value of $\WIT{\PP}[a,b]$ in $O(1)$ time and for each $b$, assuming that $\ZZ{\PP}{a}$ is already computed.
\end{lemma}
\begin{proof}	
 For an offset $(a, b)$ with $b \geq 0$, let us consider $z_{a,b} = \ZZ{\PP}{a}[b \cdot (m - a) + 1]$. 
 \begin{description}[topsep=5pt,itemsep=3pt,partopsep=0pt]
 	\item[Case 1] $z_{a,b} = (m-a)\!\cdot\!(m-b)$: Note that the value is equal to the number of elements in the overlap region. Then $\PP[1:m-a,1:m-b] \approx \PP[a+1:m,b+1:m]$, so that no witness pair exists for the offset $(a,b)$.
 	\item[Case 2] $z_{a,b} < (m-a)\!\cdot\!(m-b)$: There exists a witness pair $\pair{(i_x, i_y)}{(j_x, j_y)}$, where $(j_x, j_y)$ is the location of the element in $\PP$, that corresponds to the $(z_{a,b} + 1)$-th element of $P_1 = \ser{\PP[1\!:\!m-a, 1\!:\!m]}$.
 	By a simple calculation, we can obtain the values $(j_x, j_y)$ in $O(1)$ time.
 	We can also compute $(i_x, i_y)$ from $(j_x, j_y)$ in $O(1)$ time, similarly to the proof of Lemma~\ref{lemma:Z_p computation}, with the help of auxiliary arrays $\Prev{\PP,a}$ and $\Next{\PP,a}$. 
(Details are omitted.) 
 \end{description}
Symmetrically, we can compute it for $b<0$.
\end{proof}
 
 \begin{figure}[t]
 	\centering
 	\includegraphics[scale=0.4]{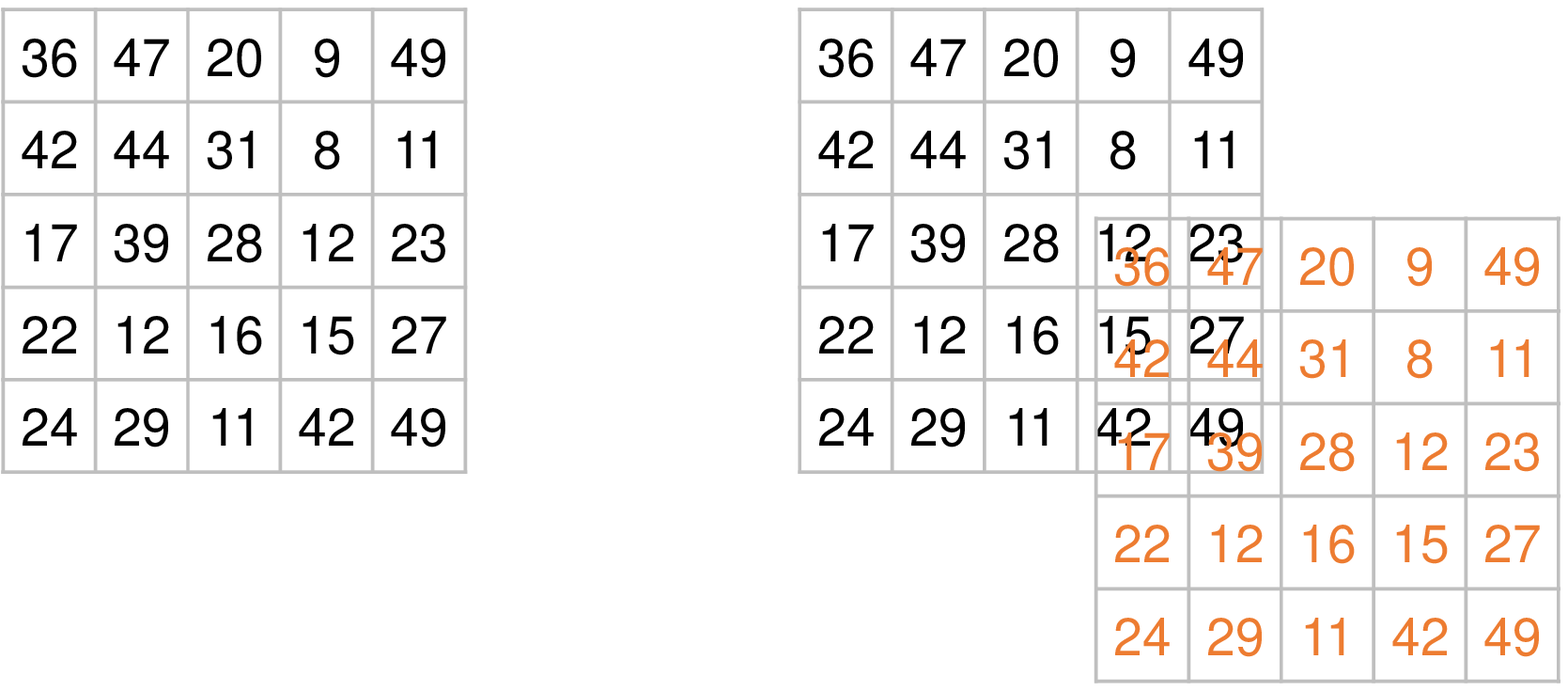}\\
 	\caption{An example of witness pair. The pattern $\PP$ is shown on the left and the alignment of $P$ with itself with offset $(3,2)$ is shown on the right. 
 	The pair $\pair{(2, 1)}{(2, 2)}$ is a witness pair for offset $(3,2)$, since $\PP[2,1] =47 > 44=\PP[2,2]$, but $\PP[5,3]=23 < 27=\PP[5,4]$.}
 	\label{fig:witness-pair}
 \end{figure}

\begin{table}[t]
	\centering
		\caption{Computation of $\ZZ{\PP}{3}$. For $\PP$ in Fig.~\ref{fig:witness-pair}, the overlap regions for offset $(3,0)$ are traversed in left-to-right/top-to-bottom order to obtain $P_1$ and $P_2$.}
		\setlength{\tabcolsep}{8pt} 	\label{table:witness-table-Z-array}
		\begin{tabular}{|c|c|c|c|c|c|c|c|c|c|c|}
			\hline
			          & $1$ & $2$ & $3$ & $4$ & $5$ & $6$ & $7$ & $8$ & $9$ & $10$ \\
			\hline
			$P_1$ & $36$ & $47$ & $42$ & $44$ & $17$ & $39$ & $22$ & $12$ & $24$ & $29$ \\
			\hline
			$P_2$ & $9$ & $49$ & $8$ & $11$ & $12$ & $23$ & $15$ & $27$ & $42$ & $49$ \\
			\hline
			$\ZZ{\PP}{3}$ & $2$ & $1$ & $2$ & $2$ & $3$ & $1$ & $2$ & $2$ & $2$ & $1$ \\
			\hline
		\end{tabular}
	\centering
\end{table}

\begin{table}[t]
	\label{table:witness-table}
	\centering
		\caption{Witness pairs for offsets $(3,0)$, $(3,1)$, $(3,2)$, $(3,3)$, $(3,3)$ for $\PP$ in Fig.~\ref{fig:witness-pair}.}
		\setlength{\tabcolsep}{3pt}
		\begin{tabular}{|c|c|c|c|c|c|}
			\hline
			$(a,b)$ & $(3,0)$ & $(3,1)$ & $(3,2)$ & $(3,3)$ & $(3,4)$\\
			\hline
			$z_{a,b}$ & $2$ & $2$ & $3$ & $2$ & $2$ \\
			\hline
			$\WIT{\PP}[a,b]$ & $\pair{(1,1)}{(2,1)}$ & $\pair{(1,2)}{(2,1)}$ & $\pair{(2,1)}{(2,2)}$ & $\pair{(1,2)}{(2,1)}$ & $\pair{(5,5)}{(5,5)}$ \\
			\hline
		\end{tabular}	
	\centering
\end{table}

\begin{lemma} \label{lemma:WIT 2d}
	We can construct the witness table $\WIT{\PP}$ in $O(m^3)$ time.
\end{lemma}

\begin{proof}
	Assume that we sorted all elements of $\PP$.
	For an arbitrarily fixed $a$, calculation of $\Prev{\PP,a}$ and $\Next{\PP,a}$ takes $O(m^2)$ time
	by using sorted $\PP$.
	$\ZZ{\PP}{a}$ can be constructed in $O(m^2)$ time by Lemma~\ref{lemma:Z-array}.
	Furthermore, finding witness pairs for all offsets $(a, b)$ takes $O(m)$ time by Lemma~\ref{lemma:Wit-2d, each cell}.
	Since there are $m$ such $a$'s to consider, $\WIT{\PP}$ can be constructed in $O(m^3)$ time.
\end{proof}

\subsection{Dueling stage}

Similarly to Lemma~\ref{lemma:consistency 1d}, we can show the transitivity as follows.
\begin{lemma}\label{lem:2D-consistency}
	For any $a , b, a', b' \geq 0$, let us consider three candidates $\TT_1 = \TT_{x,y}$, $\TT_2 = \TT_{x + a, y + b}$, and $\TT_3 = \TT_{x + a', y + b'}$. If $\TT_1$ is consistent with $\TT_2$ and $\TT_2$ is consistent with $\TT_3$, then $\TT_1$ is consistent with $\TT_3$.
\end{lemma}

The dueling algorithm due to Amir \etal~\cite{amir1994alphabet} is also applicable to the problem.
\begin{lemma} (\cite{amir1994alphabet}) \label{lemma:duel 2D}
	The dueling stage can be done in $O(n^2)$ time by using $\WIT{\PP}$.
\end{lemma}

\subsection{Sweeping stage} \label{sec:2DSweeping}
This is the hardest part for two-dimensional strings.
We first consider two surviving candidates $\TT_{x, y_1}$ and $\TT_{x,y_2}$ in some column $x$, with $y_1 < y_2$.
If we traverse $\TT[x\!:\!x+m-1, 1\!:\!n]$ from top-to-bottom/left-to-right manner we can reduce the problem to one-dimensional order-preserving problem.
Thus performing the sweeping stage for some column $x$ will take $O(nm)$ time.
Since there are $n - m -1$ such columns, the sweeping stage will take $O(n^2m)$ time.

\begin{figure}[t]
	\centering
	\begin{minipage}[t]{0.19\hsize}
		\centering
		\includegraphics[scale=0.38]{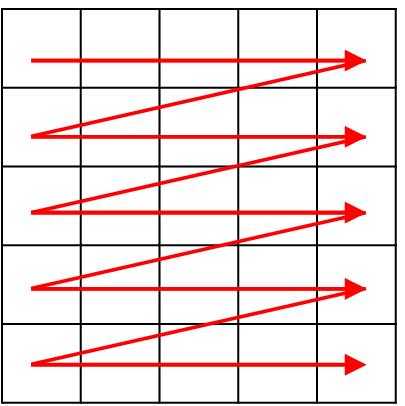}
	\end{minipage}
	\begin{minipage}[t]{0.19\hsize}
		\centering
		\includegraphics[scale=0.38]{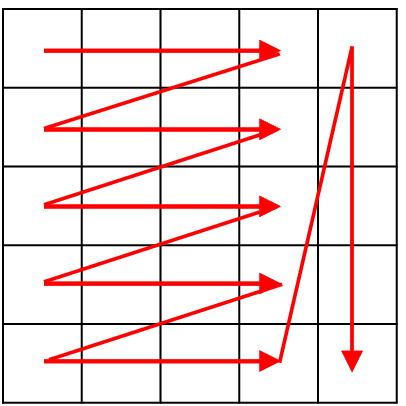}
	\end{minipage}
	\begin{minipage}[t]{0.19\hsize}
		\centering
		\includegraphics[scale=0.38]{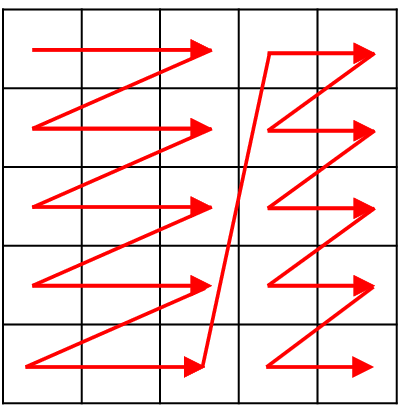}
	\end{minipage}
	\begin{minipage}[t]{0.19\hsize}
		\centering
		\includegraphics[scale=0.38]{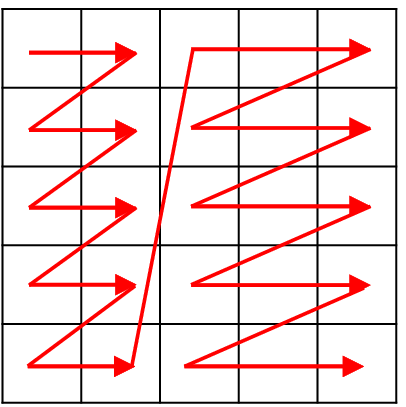}
	\end{minipage}
	\begin{minipage}[t]{0.19\hsize}
		\centering
		\includegraphics[scale=0.38]{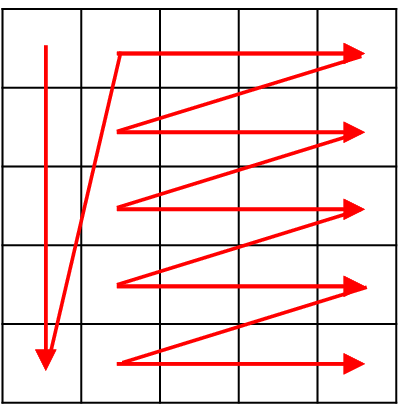}
	\end{minipage}
\caption{Example of traversing directions that we use for sweeping algorithm.}
\label{fig:serial}
\end{figure}

\begin{figure}[t]
	\centering
	\begin{minipage}[t]{0.49\hsize}
		\centering
		\includegraphics[scale=0.4]{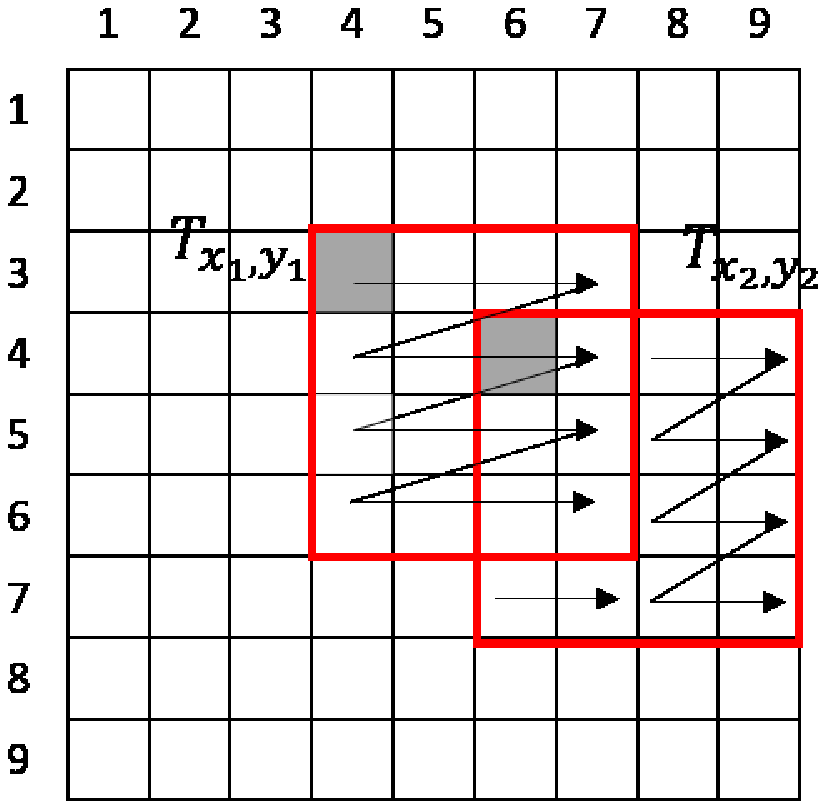}\\
		\ \ \ \scriptsize{(a)}
	\end{minipage}
	\begin{minipage}[t]{0.49\hsize}
		\centering
		\includegraphics[scale=0.4]{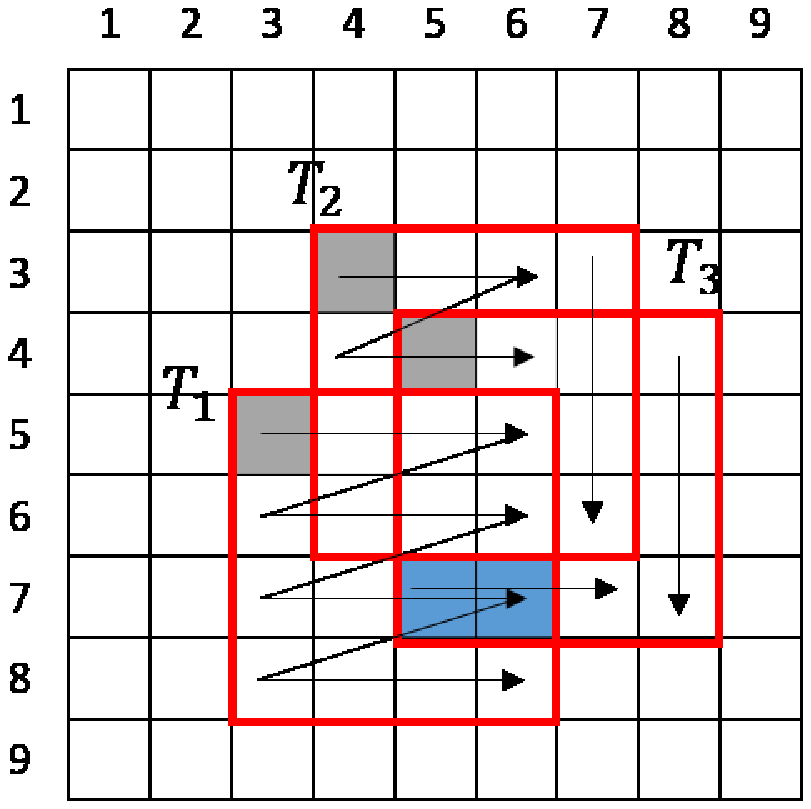}\\
		\ \ \  \scriptsize{(b)}
	\end{minipage}
	\vspace{-2mm}
	\caption{(a) Elements in the overlap region is checked only once. (b) Elements in the blue region must be checked twice.}
	\label{fig:2-3 cands}
	\vspace{-2mm}
\end{figure}

Next, we propose a method that takes advantage of consistency relation in both horizontal and vertical directions.
First, we construct $m$ strings $P_i = \ser{\PP[1:m-i,1:m]}\ser{\PP[m-i+1:m,1:m]}$ for $0 \le i < m$ by serializing $\PP$ in different way. 
We then compute $\Prev{P_i}$ and $\Next{P_i}$ for $0 \le i < m$,
thus we can compare the order-isomorphism of the pattern with the text in several different ways.
$\Prev{P_i}$ and $\Next{P_i}$ for $0 \le i < m$ can be computed in $O(n^3)$ time by sorting $\ser{\PP}$ once
and then calculated $\Prev{P_i}$ and $\Next{P_i}$ by using the sorted $\ser{\PP}$.
Fig.~\ref{fig:serial} shows $P_i$ for $0 \le i < m$ where $m=5$.
We also do the same computation for bottom-to-top/left-to-right traversing direction.

Let us consider two overlapping candidates $\TT_{x_1, y_1}$ and $\TT_{x_2, y_2}$, where $x_1 < x_2$ and $y_1 < y_2$. 
Suppose that $\TT_{x_1, y_1}$ is order-isomorphic with the pattern and we need to check $\TT_{x_2, y_2}$. Since $\TT_{x_1, y_1}$ is consistent with $\TT_{x_2, y_2}$, we need to check the order-isomorphishm of the region of $\TT_{x_2, y_2}$ that is not an overlap region. 
We do this by using $P_j$, where $j= x_2-x_1$, without checking the overlap region.
This idea is illustrated in Figure~\ref{fig:2-3 cands} (a).
The procedure for $y_1 > y_2$ is symmetrical.

Next, consider three overlapping candidates $\TT_1 = \TT_{x_1,y_1}$, $\TT_2 = \TT_{x_2,y_2}$ and $\TT_3 = \TT_{x_3,y_3}$,
such that $x_1 \leq x_2 \leq x_3$ and $y_2 \leq y_3$.
We assume that $\TT_1$ and $\TT_2$ are both order-isomorphic with the pattern.
If $y_1 \leq y_2$, we can use the method for two overlapping candidates that we described before to perform sweeping efficiently.
However, if $y_1 \geq y_2$, as showed in Fig.~\ref{fig:2-3 cands} (b),
we need to check the blue region twice since we do not know the order-isomorphism relation
between the blue region with the overlap region of $\TT_2$ and $\TT_3$.

By using the above method, we can reduce the number of comparisons for sweep stage.
However, the time complexity remains the same.
\begin{lemma} \label{lemma:2d sweep}
	The sweeping stage can be completed in $O(n^2m)$ time.
\end{lemma}

By Lemmas~\ref{lemma:WIT 2d}, \ref{lemma:duel 2D}, and \ref{lemma:2d sweep}, we conclude this section as follows.
\begin{theorem} \label{thm:duel2D}
	The duel-and-sweep algorithm solves 2d-OPPM Problem in $O(n^2 m + m^3)$ time.
\end{theorem}

%-*- coding: utf-8 -*-

\section{Discussion} \label{sec:conclusion}
In the current status,
the time complexity of duel-and-sweep algorithm for 2d-OPPM problem in Theorem~\ref{thm:duel2D} is not better than straightforward reduction to 1d-OPPM problem explained in Theorem~\ref{thm:2dOPPM}.
We showed this result as a preliminary work on solving 2d-OPPM,
and we hope the 2d-OPPM can be solved more efficiently
by finding more sophisticated method based on some unknown combinatorial properties,
as Cole \etal~\cite{Cole2014TwoDimParaMatch} did for two dimensional parameterized matching problem. 
This is left for future work.

\bibliographystyle{abbrv}
\bibliography{ref}

\end{document}